\documentclass{article}
\pdfoutput=1
\usepackage[a4paper, margin=3cm, bottom=4cm]{geometry}
\usepackage[utf8]{inputenc}
\usepackage[T1]{fontenc}
\usepackage{lmodern}
\usepackage[
    backend=biber,
    style=alphabetic, maxalphanames=5,
    maxbibnames=100, maxsortnames=100,
    sorting=nyt,
    natbib=true
]{biblatex}
\usepackage{enumitem}
\usepackage[bookmarks, bookmarksnumbered,
			colorlinks=true,
            linkcolor = blue,
            urlcolor  = blue,
            citecolor = teal]{hyperref}
\usepackage{myquicksetup}
\usepackage{mylivemacros}

\title{The entropy power conjecture implies the McKean conjecture}
\date{August 26, 2024}

\author{Guillaume Wang}

\hypersetup{
    pdftitle={The entropy power conjecture implies the McKean conjecture},
    pdfauthor={Guillaume Wang}
}

\begin{document}
\maketitle

\begin{abstract}
    After reviewing the entropy power, the McKean, and the Gaussian completely monotone conjectures,
    we prove that the first implies the second, for each order of the time-derivative.
    The proof is elementary and is based on manipulating the Bell polynomials.
\end{abstract}
\vspace{1em}

Let us review three conjectures about the successive time-derivatives of the entropy along the heat flow, following \cite{ledoux2022differentials}.
Let, for some probability measure $\mu_0$ over $\RR^D$,
$\mu_t = \mathrm{Law}(X_0 + \sqrt{t} G)$ where $G \sim \NNN(0, I_D)$ and $X_0 \sim \mu_0$. Equivalently, $\partial_t \mu_t = \frac12 \Delta \mu_t$, the heat equation.
Let
\begin{equation*}
	y(t) = -2 H(\mu_t)/D,
    \qquad \quad
	\dot{y}(t) = I(\mu_t)/D,
    \qquad \quad
	N(t) = e^{y(t)}
    = e^{-\frac{2}{D} H(\mu_t)},
\end{equation*}
where $H(\mu) = \int d\mu \log \frac{d\mu}{dx}$ is the (negative) differential entropy
and 
$I(\mu) = \int d\mu \norm{\nabla \log \mu}^2$ is the Fisher information.
Note that De Bruijn's identity asserts that 
$\frac{d}{dt} y(t) = \dot{y}(t)$.
$N(t)$ is called the entropy power of $\mu_t$.

Suppose $\mu_0$ has covariance $\sigma^2 I_D$ and denote $\sigma_t^2 = \sigma^2 + t$.
\begin{itemize}
	\item The entropy power conjecture \cite{toscani2015concavity} states that
	for all $m \geq 1$, 
    \begin{equation} \label{eq:entropy_power} \tag{EP}
        \forall t>0,~ (-1)^{m-1} \frac{d^m}{dt^m} N(t) \geq 0.
    \end{equation}
	\item The McKean conjecture \cite[Section~12]{mckean1966speed} states that
	for all $m \geq 1$,
    \begin{equation} \label{eq:mckean} \tag{McK}
        \forall t>0,~ (-1)^{m-1} \frac{d^{m-1}}{dt^{m-1}} \dot{y}(t) \geq \frac{(m-1)!}{(\sigma^2_t)^m}.
    \end{equation}
	\item The Gaussian completely monotone conjecture \cite{cheng2015higher} states that 
 	for all $m \geq 1$,
    \begin{equation} \label{eq:GCM} \tag{GCM}
        \forall t>0,~ (-1)^{m-1} \frac{d^{m-1}}{dt^{m-1}} \dot{y}(t) \geq 0.
    \end{equation}
\end{itemize}
\eqref{eq:entropy_power} and \eqref{eq:mckean} hold true for any $D$ when $m \leq 3$ and $\mu_0$ is log-concave \cite{toscani2015concavity},
\eqref{eq:mckean} holds true when $D=1$ and $m \leq 5$ and $\mu_0$ is log-concave \cite{zhang2018gaussian},
and \eqref{eq:GCM} holds true when
$D \leq 2$ and $m\leq 4$ \cite{cheng2015higher,guo2022lower},
as well as when $D \leq 4$ and $m \leq 3$ \cite{guo2022lower}.

The purpose of this note is to prove the following implication, which was hinted at in \cite[end of Section~4: ``It might be that the Entropy Power Conjecture is stronger than the McKean Conjecture'']{ledoux2022differentials}.
(This note does not imply any new result compared to the state of the art.)
\begin{proposition}
	For any $M \in \NN$, if \eqref{eq:entropy_power} holds for all $m \leq M$, then \eqref{eq:mckean} holds for all $m \leq M$.
\end{proposition}

\begin{proof}
    We denote by $B_n(X_1, ..., X_n)$ or $B_n(X_1, X_2, ...)$ the complete exponential Bell polynomials.
    By Faa di Bruno's formula, for any $m \geq 1$,
    \begin{equation*}
        \frac{d^m}{dt^m} N(t) = \frac{d^m}{dt^m} e^{y(t)}
    	= e^{y(t)} B_m(\dot{y}, \dot{y}', \dot{y}'', ...)
    \end{equation*}
    where $\dot{y}' = \frac{d}{dt} \dot{y}$, $\dot{y}'' = \frac{d^2}{dt^2} \dot{y}$, etc.
    Moreover, by property of the Bell polynomials,
    \begin{equation*}
    	\forall n,~ \forall \beta \in \RR,~ B_n(\beta X_1, \beta^2 X_2, \beta^3 X_3, ...) = \beta^n B_n(X_1, X_2, ...).
    \end{equation*}
    So for $\beta=-1$, letting $Y_k = (-1)^k \dot{y}^{(k-1)}$ for all $k \geq 1$,
    \begin{align*}
    	(-1)^{m-1} \frac{d^m}{dt^m} N(t) 
        = (-1)^{m-1} e^{y(t)} B_m(\dot{y}, \dot{y}', \dot{y}'', ...) 
        &= -e^{y(t)} B_m(-\dot{y}, \dot{y}', -\dot{y}'', \dot{y}''', ...) \\
        &= -e^{y(t)} B_m(Y_1, Y_2, ...).
    \end{align*}
    
    Fix $M \in \NN$ and suppose that \eqref{eq:entropy_power} holds for all $1 \leq m \leq M$, i.e., $B_m(Y_1, Y_2, ...) \leq 0$ for all $1 \leq m \leq M$.
    Then by \autoref{lm:first_lm} below,
    \begin{align*}
        Y_m &\leq -(m-1)! (-Y_1)^m ~~~~\text{for all}~~ 1 \leq m \leq M, \\
        \text{i.e.,}~~~~
        (-1)^{m-1} \dot{y}^{(m-1)} &\geq (m-1)!~ \dot{y}^m.
    \end{align*}
    Now by the Cram\'er-Rao lower bound, $\dot{y}(t) = I(\mu_t)/D \geq \frac{1}{\sigma_t^2}$,
    hence the inequality \eqref{eq:mckean}.
\end{proof}

\begin{lemma} \label{lm:first_lm}
	Let $N \in \NN$ and $Y_1, Y_2, ... \in \RR$
    such that $B_n(Y_1, Y_2, ...) \leq 0$ for all $1 \leq n \leq N$.
	Then $Y_1 \leq 0$ and $Y_n \leq -(n-1)! (-Y_1)^n \leq 0$ for all $1 \leq n \leq N$.
\end{lemma}

\begin{proof}
	We proceed by finite induction over $n$.
	The case $n=1$ is clear as $B_1(Y_1, Y_2, ...) = Y_1$.
	Fix $n$, suppose $Y_k \leq -(k-1)! (-Y_1)^k \leq 0$ for all $k \leq n$, and let us show the inequality for $n+1$.
	The Bell polynomials satisfy the recurrence relation
	\begin{gather*}
		B_{n+1}(Y_1, Y_2, ...) = \sum_{i=0}^{n} \binom{n}{i} B_{n-i}(Y_1, Y_2, ...) Y_{i+1} \leq 0, \\
		\text{so}~~~~ Y_{n+1} = B_0(Y_1, Y_2, ...) Y_{n+1} \leq -\sum_{i=0}^{n-1} \binom{n}{i} B_{n-i}(Y_1, Y_2, ...) Y_{i+1}.
	\end{gather*}
	Since $B_{n-i}(Y_1, Y_2, ...) \leq 0$ and $-Y_{i+1} \geq 0$ for all $i \leq n-1$ by induction hypothesis, then
	\begin{equation*}
		Y_{n+1} \leq \sum_{i=0}^{n-1} \binom{n}{i} B_{n-i}(Y_1, Y_2, ...) (-Y_{i+1})
        \leq 0 + \binom{n}{n-1} B_1(Y_1, Y_2, ...) (-Y_n)
        = n Y_1 (-Y_n).
	\end{equation*}
    So since $Y_1 \leq 0$ and $Y_n \leq -(n-1)! (-Y_1)^n$ by induction hypothesis,
	\begin{equation*}
		Y_{n+1} 
        \leq n (-Y_1) Y_n
        \leq n (-Y_1) \left[ -(n-1)! (-Y_1)^n \right]
		= -n! (-Y_1)^{n+1},
	\end{equation*}
	which concludes the proof by induction.
\end{proof}

\printbibliography

\end{document}